\documentclass[10pt]{article}

\usepackage[utf8]{inputenc}  
\usepackage{amsmath, amsthm, amssymb}  
\usepackage{algorithm, algpseudocode}  

\usepackage{geometry}  
\usepackage{titlesec}  
\usepackage{lipsum}  
\usepackage{graphicx} 
\usepackage{multicol} 
\usepackage{balance}  
\graphicspath{{images/}}
\usepackage{subcaption}
\usepackage{booktabs}

\usepackage{microtype}
\usepackage{tikz}

\usepackage[numbers,sort&compress]{natbib}
\titleformat{\section}[block]
  {\normalfont\fontsize{12}{14}\bfseries}{\thesection}{1em}{}

\titleformat{\subsection}[block]
  {\normalfont\fontsize{11}{13}\bfseries}{\thesubsection}{1em}{}

\titleformat{\subsubsection}[block]
  {\normalfont\small\bfseries}{\thesubsubsection}{1em}{}

\newtheoremstyle{tight}
  {10pt}                     
  {10pt}                     
  {}                         
  {}                         
  {\bfseries}                
  {.}                        
  { }                        
  {}                         
\theoremstyle{tight}

\newtheorem{definition}{Definition}[section]
\newtheorem{lemma}{Lemma}[section]
\newtheorem{theorem}{Theorem}[section]
\newtheorem{corollary}{Corollary}[section]

\newtheoremstyle{tightproof}
  {10pt}                     
  {10pt}                     
  {}                         
  {}                         
  {\bfseries}                
  {.}                        
  { }                        
  {}                         
\theoremstyle{tightproof}
\renewenvironment{proof}[1][\proofname]{\par\noindent{\bfseries #1. }}{\qed\par}

\title{\textbf{A Polynomial-Time Algorithm for Computing the Exact Convex Hull in High-Dimensional Spaces} \\ \vspace{10pt}}

\author{
    Qianwei ZHUANG\textsuperscript{*} \\[0.5em]
    Research Center for Advanced Science and Technology,
    The University of Tokyo, Japan \\[0.5em]
    \textsuperscript{*}\texttt{qianweizhuang@g.ecc.u-tokyo.ac.jp; qweizhuang@gmail.com} \\
}

\date{} 

\begin{document}

\twocolumn
\maketitle

\begin{abstract}
This study presents a novel algorithm for identifying the set of extreme points that constitute the exact convex hull of a point set in high-dimensional Euclidean space. The proposed method iteratively solves a sequence of dynamically updated quadratic programming (QP) problems for each point and exploits their solutions to provide theoretical guarantees for exact convex hull identification. For a dataset of \( n \) points in an \( m \)-dimensional space, the algorithm achieves a dimension-independent worst-case time complexity of \( O(n^{p+2} \log(1/\epsilon)) \), where \( p \) depends on the choice of QP solver (e.g., \( p = 4 \) corresponds to the worst-case bound when using an interior-point method), and \( \epsilon \) denotes the target numerical precision (i.e., the optimality tolerance of the QP solver).

The proposed method is applicable to spaces of arbitrary dimensionality and exhibits particular efficiency in high-dimensional settings, owing to its polynomial-time complexity, whereas existing exponential-time algorithms become computationally impractical.
\end{abstract}

\vspace{1em} 
\vspace{1em}

\section{Introduction}
In geometry, the convex hull of a set of points in \( \mathbb{R}^m \) is the smallest convex set that encloses all of the points. Identifying the vertices, also known as extreme points, that characterize the convex hull is a fundamental computational problem with broad applications across many fields \cite{de2000computational}, such as machine learning \cite{nemirko2021machine}, mathematical programming \cite{wei2024convex}, and computer graphics \cite{wei2022approximate}, among others \cite{sivakumar2022comparative,jayaram2016convex,chakravarthy1998obstacle}.

Efficient algorithms for computing the convex hull in low-dimensional spaces (i.e., when \( m = 2 \) or \( 3 \)) have been developed with a worst-case time complexity of \( O(n \log n) \) \cite{graham1972efficient, chand1970algorithm, jarvis1973identification} \footnote{A comprehensive overview of various algorithms is available in \cite{ebert2014interpolation}}. For higher-dimensional spaces, notable methods include the Clarkson-Shor algorithm \cite{clarkson1988applications} and Quickhull \cite{barber1996quickhull}, both exhibiting a worst-case complexity of \( O(n^{\lfloor m/2 \rfloor}) \). This exponential growth in complexity with respect to the dimension $m$ presents significant challenges for both computing and applying convex hulls.

In contrast to exact methods, approximate convex hull alternatives have been extensively explored to address computational challenges in high-dimensional settings. Early approaches include grid-based sampling techniques \cite{bentley1982approximation} and polar-coordinate grid methods \cite{xu1998approximate}, which discretize the space to approximate the hull. Norm-based approaches were later developed \cite{wang2013online, khosravani2013simple}, though they do not explicitly quantify the quality of approximation. Later, a greedy algorithm introduced in \cite{sartipizadeh2016computing} offers a provable guarantee on solution quality. Further advances address robustness to noise, such as the method in \cite{awasthi2018robust}, which can exactly recover the convex hull under provided a suitable robustness parameter $\gamma$ is supplied in noiseless cases. More recently, a neural network-based framework has been proposed for convex hull approximation using deep convex architectures \cite{balestriero2022deephull}.

Although approximation methods improve computational efficiency, they often incur a loss in precision. This work presents an algorithm that computes the exact convex hull in high-dimensional spaces with polynomial time complexity.

The main contributions of this study are summarized as follows. A procedure is proposed to iteratively construct a compact reference set consisting of extreme points, enabling the identification of whether a given point is an extreme point. The extreme points constituting the reference set correspond to the independent decision variables of the associated QP problem, whose objective value asymptotically approaches zero as additional extreme points are incrementally incorporated into the set. Theoretical analysis confirms that the proposed iterative mechanism ensures the optimal objective value of each point’s associated QP converges to zero within a finite number of iterations, thereby yielding the set of extreme points that constitute the exact convex hull once convergence is achieved for all points.

The remainder of this paper is organized as follows. 
Section~2 introduces the notation and reviews the essential theoretical foundations. 
Section 3 presents the proposed method for identifying the set of extreme points that represent the exact convex hull.
Section~4 provides an analysis of its computational complexity. 
Finally, Section~5 concludes the paper with a summary of the key findings.

\section{Convex Hull Problem}
Let \(\mathcal{A} \subset \mathbb{R}^m\) denote a finite set of \(n\) distinct points, expressed as
\(
\mathcal{A} = \{\boldsymbol{x}_i : i \in \mathcal{I}_\mathcal{A}\},
\)
where \(\mathcal{I}_\mathcal{A}\) is an index set of cardinality \(|\mathcal{I}_\mathcal{A}| = n\). The Euclidean norm of \(\boldsymbol{x}_i \in \mathcal{A}\) is denoted by \(\|\boldsymbol{x}_i\|\), and its \(j\)-th component is written as \(x_{ji}\) for \(j = 1, \dots, m\). For any subset \(\mathcal{S} \subseteq \mathcal{A}\), the set difference \(\mathcal{A} \setminus \mathcal{S}\) represents the elements of \(\mathcal{A}\) excluding those in \(\mathcal{S}\). 

This study seeks to identify the vertices of the polytope defined by $\mathcal{A}$, which constitute its convex hull, denoted by $\texttt{conv}(\mathcal{A})$. These vertices are the extreme points of $\mathcal{A}$.

\begin{definition}
\label{def:extreme point}
    $\boldsymbol{x}_l \in \mathcal{A}$ is an extreme point of $\mathcal{A}$ if and only if it cannot be expressed as a convex combination of points from $\mathcal{A} \setminus \{\boldsymbol{x}_l\}$.
\end{definition}

Let \(\mathcal{E}\) be the set of extreme points of \(\mathcal{A}\), thereby \(\texttt{conv}(\mathcal{A}) = \texttt{conv}(\mathcal{E})\). Lemma~\ref{lemma:extreme point} \cite{Tropp2018} provides a practical criterion for identifying an extreme point (with \(\langle \cdot, \cdot \rangle\) denoting the inner product in \(\mathbb{R}^m\)):

\begin{lemma}
\label{lemma:extreme point}
Assume \( \boldsymbol{v} \in \mathbb{R}^m \setminus \{ \boldsymbol{0} \} \). For a point \( \boldsymbol{x}_l \in \mathcal{A} \), if it uniquely satisfies
\begin{equation}
\label{condition:argmax}
    \boldsymbol{x}_l = \operatorname*{\textit{arg}\,\textit{max}}_{i \in \mathcal{I}_\mathcal{A}} \langle \boldsymbol{x}_i, \boldsymbol{v} \rangle,
\end{equation}
then \( \boldsymbol{x}_l \) is an extreme point of \(\mathcal{A}\).
\end{lemma}

An alternative proof of Lemma~\ref{lemma:extreme point} is provided in Appendix~\ref{lemmaproof}, 
showing that any $\boldsymbol{x}_l \in \mathcal{A}$ satisfying condition~\eqref{condition:argmax} cannot be represented as a convex combination of the other points in 
\(\mathcal{A} \setminus \{\boldsymbol{x}_l\}\).

Any subset \(\mathcal{S} \subseteq \mathcal{A}\) induces a convex hull, denoted by \(\texttt{conv}(\mathcal{S})\), which satisfies \(\texttt{conv}(\mathcal{S}) \subseteq \texttt{conv}(\mathcal{A})\). To define this convex hull, consider the Euclidean distance from a point \(\boldsymbol{z} \in \mathbb{R}^m\) to \(\texttt{conv}(\mathcal{S})\), denoted by \(d(\boldsymbol{z}, \mathcal{S})\). The squared distance is obtained by solving the QP problem defined in Model~\eqref{model:qp}, where $\mathcal{R} \subseteq \mathcal{A}$ denotes the reference set of points.

\begin{equation}
\label{model:qp}
\begin{aligned}
d(\boldsymbol{z}, \mathcal{R})^2
= & \min_{\lambda_i}
\left\| \boldsymbol{z} - \sum_{i \in \mathcal{I}_\mathcal{R}} \lambda_i \boldsymbol{x}_i \right\|^2 \\
& \text{s.t. } \sum_{i \in \mathcal{I}_\mathcal{R}} \lambda_i = 1,\;
\lambda_i \ge 0 \; \forall i
\end{aligned}
\end{equation}

Model~\eqref{model:qp} is equivalent to projecting the point $\boldsymbol{z}$ onto $\texttt{conv}(\mathcal{R})$. Projection onto a convex set is a convex optimization problem \cite{boyd2004convex}. Several approaches can be employed to solve Model~\eqref{model:qp}, such as interior-point methods~\cite{boyd2004convex}, which typically have a time complexity\footnote{Complexity varies by method; see relevant studies on QP solution methods, such as \cite{cai2013complexity, boudjellal2022complexity, wu2025eiqp}, for details.} of $O(|\mathcal{R}|^{p} \log(1/\epsilon))$, where $p = 4$ corresponds to the worst-case bound~\cite{ye1989extension, boyd2004convex} and $\epsilon$ denotes the target solution accuracy.

Let \( \lambda_i^{*} \) denote the optimal solution to Model~\eqref{model:qp}. The distance \( d(\boldsymbol{z}, \mathcal{R}) \) is then given by
\[
d(\boldsymbol{z}, \mathcal{R}) = \left\| \boldsymbol{z} - \sum_{i \in \mathcal{I}_\mathcal{R}} \lambda^*_i \boldsymbol{x}_i \right\|.
\]
Note that \( \boldsymbol{z} \in \texttt{conv}(\mathcal{R}) \) if and only if
\[
d(\boldsymbol{z}, \mathcal{R}) = 0.
\]
It is evident that 
\begin{equation*}
d(\boldsymbol{x}_l, \mathcal{A}) = d(\boldsymbol{x}_l, \mathcal{E}) = 0, \quad \forall l \in \mathcal{I}_\mathcal{A}.
\end{equation*}
This implies that the distance from every \( \boldsymbol{x}_l \) in \( \mathcal{A} \) to \( \texttt{conv}(\mathcal{A}) \) and \( \texttt{conv}(\mathcal{E}) \) is zero.

Table~\ref{tab:instance} presents a two-dimensional point set used for illustration in the subsequent sections. 
Each point \( \boldsymbol{x}_i \), labeled \( i = \mathsf{1}, \dots, \mathsf{9} \), is referred to as Point~\( i \). 
The set of extreme points is 
\(
\mathcal{E} = \{\boldsymbol{x}_1, \boldsymbol{x}_2, \boldsymbol{x}_3, \boldsymbol{x}_4, \boldsymbol{x}_5, \boldsymbol{x}_6, \boldsymbol{x}_9\},
\)
representing the convex hull generated by \( \mathcal{A} \), as shown in Figure~\ref{fig:complete_convexhull}.

\begin{table}[ht]
    \centering
    \caption{A Two-Dimensional Point Set}
    \begin{tabular}{c c c c c c c c c c}
        \toprule
        Point & $\mathsf{1}$  & $\mathsf{2}$  & $\mathsf{3}$  & $\mathsf{4}$  & $\mathsf{5}$  & $\mathsf{6}$ & $\mathsf{7}$ & $\mathsf{8}$  & $\mathsf{9}$  \\
        \midrule
        $x_{1}$ & 10 & 72 & 40 & 46 & 32 & 71 & 34 & 40 & 62 \\
        $x_{2}$ & 26 & 20 & 1  & 76 & 72 & 36 & 66 & 38 & 69 \\
        \bottomrule
    \end{tabular}
    \label{tab:instance}
\end{table}

\begin{figure}[ht]
    \centering
    \includegraphics[width=0.45\textwidth]{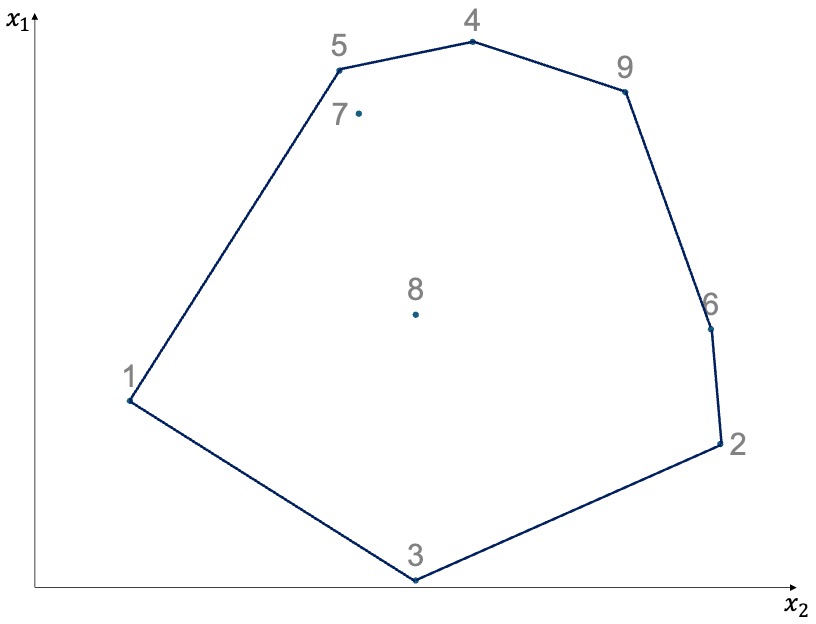}
    \caption{The Exact Convex Hull}
    \label{fig:complete_convexhull}
\end{figure}

\subsection*{Initializing a Subset of Extreme Points}
According to Lemma~\ref{lemma:extreme point}, an extreme point can be identified by applying any nonzero vector \( \boldsymbol{v} \). For instance, \( \boldsymbol{v} \) may be chosen as one of the column vectors of the identity matrix \( \mathbf{I}_m \), which enables the detection of extreme points that highlight extremal features in each dimension. Let \( I_c \) denote the \( c \)-th column vector of \( \mathbf{I}_m \). Considering the two-dimensional example in Table~\ref{tab:instance}, the corresponding extreme points can be identified as\footnote{This initialization strategy can be traced back to early works such as \cite{graham1972efficient, jarvis1973identification, barber1996quickhull}.}
\[
\boldsymbol{x}_1 = \operatorname*{\textit{arg\,max}}_{i \in \mathcal{I}_\mathcal{A}} \langle \boldsymbol{x}_i, -I_1 \rangle, \quad
\boldsymbol{x}_2 = \operatorname*{\textit{arg\,max}}_{i \in \mathcal{I}_\mathcal{A}} \langle \boldsymbol{x}_i, I_1 \rangle,
\]
\[
\boldsymbol{x}_3 = \operatorname*{\textit{arg\,max}}_{i \in \mathcal{I}_\mathcal{A}} \langle \boldsymbol{x}_i, -I_2 \rangle, \quad
\boldsymbol{x}_4 = \operatorname*{\textit{arg\,max}}_{i \in \mathcal{I}_\mathcal{A}} \langle \boldsymbol{x}_i, I_2 \rangle.
\]

Let $\mathcal{E}' \subseteq \mathcal{E}$ denote a subset of extreme points, where 
$\mathcal{E}' = \{\boldsymbol{x}_1, \boldsymbol{x}_2, \boldsymbol{x}_3, \boldsymbol{x}_4\}$ 
represents a partial convex hull, as shown in Figure~\ref{fig:instance}.

\begin{figure}[ht]
    \centering
    \includegraphics[width=0.45\textwidth]{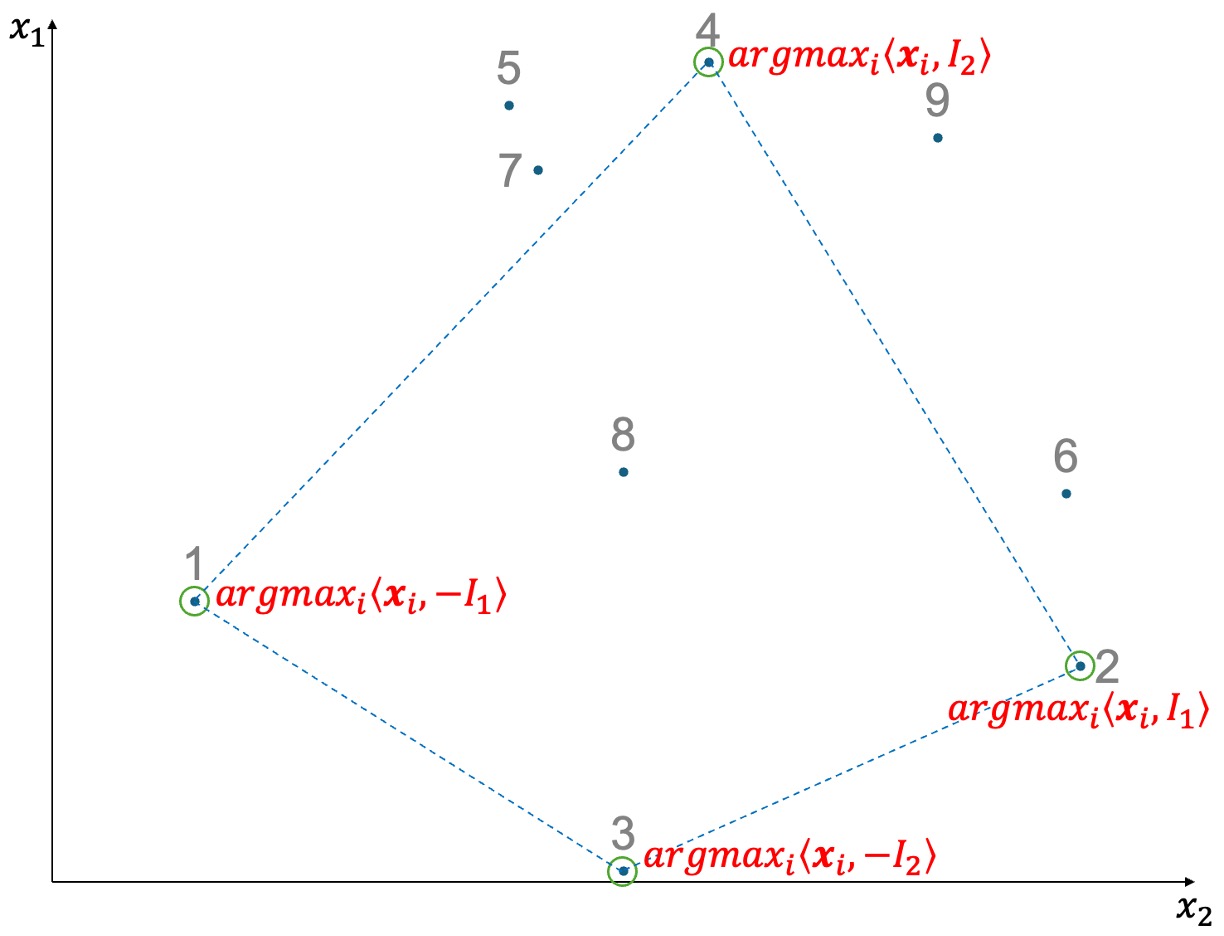}
    \caption{Partial Convex Hull Initialization}
    \label{fig:instance}
\end{figure}

Although a random selection of $\boldsymbol{v}$ can reveal a subset $\mathcal{E}' \subseteq \mathcal{E}$, determining the complete set $\mathcal{E}$ remains challenging. The following section presents the proposed approach for identifying $\mathcal{E}$ in its entirety. It should be noted that $\mathcal{E}$ is not necessarily unique.

\section{Exact Convex Hull Construction}
Solving Model~\eqref{model:qp} serves as the basis in the proposed method. Since the computational complexity depends on \( |\mathcal{R}| \), it is desirable to minimize its size. A key requirement is that \( \mathcal{R} \) must enclose the point \( \boldsymbol{x}_l \), that is, \( d(\boldsymbol{x}_l, \mathcal{R}) = 0 \). Algorithm~\ref{algorithm:ConstructingConvexHull} introduces an iterative procedure to construct a compact reference set \( \mathcal{R}_l \) for each point \( \boldsymbol{x}_l \), until the enclosure condition \( d(\boldsymbol{x}_l, \mathcal{R}_l) = 0 \) is fulfilled.

\begin{algorithm}[htbp]
\caption{Constructing Convex Hull}
\label{algorithm:ConstructingConvexHull}
\begin{algorithmic}[1]  
\Require $\mathcal{A}$
\Ensure $\mathcal{E}$
\State initialize $\mathcal{E}'$
\State $\mathcal{T} \gets \mathcal{A} \setminus \mathcal{E}'$
\While{$\mathcal{T} \neq \emptyset$}
    \State select $\boldsymbol{x}_l \in \mathcal{T}$
    \State initialize $\mathcal{R}_{l}$
    \State $d(\boldsymbol{x}_l,\mathcal{R}_l) \gets \text{solve Model~\eqref{model:qp}}$
    \While{$d(\boldsymbol{x}_l,\mathcal{R}_l) > 0$}
        \State $\boldsymbol{v}^{*} \gets \boldsymbol{x}_l - \sum_{i \in \mathcal{I}_{\mathcal{R}_l}} \lambda_i^{*}\boldsymbol{x}_i$
        \State $\boldsymbol{x}_{i^{etm}} \gets \operatorname*{\textit{arg}\,\textit{max}}_{i \in \mathcal{I}_\mathcal{A}} \langle \boldsymbol{x}_i, \boldsymbol{v}^{*} \rangle$
        \State $\mathcal{R}_l \gets \mathcal{R}_l \cup \{\boldsymbol{x}_{i^{etm}}\}$
        \State $d(\boldsymbol{x}_l,\mathcal{R}_l) \gets \text{solve Model~\eqref{model:qp}}$
        \State $\mathcal{T} \gets \mathcal{T} \setminus \{\boldsymbol{x}_{i^{etm}}\}$
    \EndWhile
    \State $\mathcal{T} \gets \mathcal{T} \setminus \{\boldsymbol{x}_l\}$
    \State $\mathcal{E}' \gets \mathcal{E}' \cup \mathcal{R}_l$
\EndWhile
\State $\mathcal{E} \gets \mathcal{E}'$
\State \Return $\mathcal{E}$
\end{algorithmic}
\end{algorithm}

Algorithm~\ref{algorithm:ConstructingConvexHull} takes as input a finite set of data points $\mathcal{A}$ and identifies its set of extreme points $\mathcal{E}$. In line~1, $\mathcal{E}'$ can be initialized using the previously introduced identity matrix $\mathbf{I}_m$ or any arbitrary non-zero vector $\boldsymbol{v}$. In line~5, $\mathcal{R}_l$ is initialized by randomly selecting a data point from $\mathcal{E}'$. These random initializations do not affect the convergence of the while loop in lines~7--13 of Algorithm~\ref{algorithm:ConstructingConvexHull}, as guaranteed by Theorem~\ref{theorem:convergence}.

\begin{theorem}
\label{theorem:convergence}
For any initial reference set \(\mathcal{R}_l \subseteq \mathcal{A}\) corresponding to \(\boldsymbol{x}_l \in \mathcal{A}\), the procedure outlined in lines 7-13 of Algorithm~\ref{algorithm:ConstructingConvexHull} ensures the value \( d(\boldsymbol{x}_l,\mathcal{R}_l) \) converges to zero.
\end{theorem}

\begin{proof}[Proof]
Let \(\mathcal{R}_l^k\) denote the reference set for \(\boldsymbol{x}_l\) at iteration \(k\), and assume the distance from \(\boldsymbol{x}_l\) to \(\texttt{conv}(\mathcal{R}_l^k)\) as  
\begin{equation*}
d(\boldsymbol{x}_l,\mathcal{R}_l^k) > 0.
\end{equation*}
Rewrite the operation in line 10 as  
\begin{equation*}
\mathcal{R}_l^{k+1} \gets \mathcal{R}_l^k \cup \{\boldsymbol{x}_{i^{etm}}\}.
\end{equation*}
According to line 8, it holds that  
\[
\boldsymbol{v}^* = \boldsymbol{x}_l - \sum_{i \in \mathcal{I}_{\mathcal{R}_l^k}} \lambda_i^* \boldsymbol{x}_i.
\]
where \(\lambda_i^*\) are the optimal coefficients solving Model~\eqref{model:qp} with \(\mathcal{R} \gets \mathcal{R}_l^k\). \\
In line~9, if \(\boldsymbol{x}_{i^{etm}} = \boldsymbol{x}_l\), then
\[
d(\boldsymbol{x}_l, \mathcal{R}_l^{k+1}) = 0,
\]
which guarantees termination of the while loop.
Otherwise, if \(\boldsymbol{x}_{i^{etm}} \ne \boldsymbol{x}_l\), since \(\boldsymbol{x}_{i^{etm}}\) uniquely satisfies
\[
\boldsymbol{x}_{i^{etm}} = \underset{i \in \mathcal{I}_\mathcal{A}}{\arg\max} \, \langle \boldsymbol{v}^{*}, \boldsymbol{x}_i \rangle,
\]
it holds that 
\begin{equation}
\boldsymbol{v}^{*T} \boldsymbol{x}_{i^{etm}} > \boldsymbol{v}^{*T} \boldsymbol{x}_i, \quad \forall i \in \mathcal{I}_{\mathcal{R}_l^k}.
\label{eq:vstar_xetm}
\end{equation}
Multiplying both sides of \eqref{eq:vstar_xetm} by \(\lambda_i^*\) and summing over all \(i \in \mathcal{I}_{\mathcal{R}_l^k}\) yields  
\[
\sum_{i \in \mathcal{I}_{\mathcal{R}_l^k}} \lambda_i^* \boldsymbol{v}^{*T} \boldsymbol{x}_{i^{etm}} > \sum_{i \in \mathcal{I}_{\mathcal{R}_l^k}} \lambda_i^* \boldsymbol{v}^{*T} \boldsymbol{x}_i.
\]
Using the fact that \(\sum_{i \in \mathcal{I}_{\mathcal{R}_l^k}} \lambda_i^* = 1\), it follows that  
\[
\boldsymbol{v}^{*T} \boldsymbol{x}_{i^{etm}} > \boldsymbol{v}^{*T} \sum_{i \in \mathcal{I}_{\mathcal{R}_l^k}} \lambda_i^* \boldsymbol{x}_i.
\]  
Rearranging terms yields  
\begin{equation}
\boldsymbol{v}^{*T} \left( \boldsymbol{x}_{i^{etm}} - \sum_{i \in \mathcal{I}_{\mathcal{R}_l^k}} \lambda_i^* \boldsymbol{x}_i \right) > 0.
\label{eq:vstar_vsharp}
\end{equation}
Define  
\[
\boldsymbol{x}_{i^*} = \sum_{i \in \mathcal{I}_{\mathcal{R}_l^k}} \lambda_i^* \boldsymbol{x}_i,
\]
and  
\[
\boldsymbol{v}^{\#} = \boldsymbol{x}_{i^{etm}} - \boldsymbol{x}_{i^*}.
\]
With these definitions, inequality \eqref{eq:vstar_vsharp} can be rewritten as  
\[
\langle \boldsymbol{v}^*, \boldsymbol{v}^\# \rangle > 0.
\]
Let \(\theta\) denote the angle between vectors \(\boldsymbol{v}^{*}\) and \(\boldsymbol{v}^{\#}\), where  
\begin{equation*}
0 \leq \theta \leq \pi.
\end{equation*}
The points \(\boldsymbol{x}_l\), \(\boldsymbol{x}_{i^*}\), and \(\boldsymbol{x}_{i^{etm}}\) form a triangle. Using the relation  
\[
\langle \boldsymbol{v}^*, \boldsymbol{v}^\# \rangle = \|\boldsymbol{v}^*\| \|\boldsymbol{v}^\#\| \cos \theta > 0,
\]
it follows that the angle \(\theta\) satisfies  
\[
0 < \sin \theta < 1.
\]
Define  
\[
\boldsymbol{x}_{i'} = \sum_{i \in \{i^*, i^{etm}\}} \lambda_i^{\prime *} \boldsymbol{x}_i,
\]
where \(\lambda_i^{\prime *}\) denotes the solution to Model~\eqref{model:qp} with \(\mathcal{R} \gets \{\boldsymbol{x}_{i^*}, \boldsymbol{x}_{i^{etm}}\}\). It follows that \(\boldsymbol{x}_{i'} \in \texttt{conv}(\mathcal{R}_l^{k+1})\). Define \(\boldsymbol{v}' = \boldsymbol{x}_l - \boldsymbol{x}_{i'}\). Then,  
\[
\|\boldsymbol{v}'\| = \|\boldsymbol{v}^*\| \sin \theta < \|\boldsymbol{v}^*\|.
\]
Further, define  
\[
\boldsymbol{x}_{i^{**}} = \sum_{i \in \mathcal{I}_{\mathcal{R}_l^k}} \lambda_i^{**} \boldsymbol{x}_i,
\]
where \(\lambda_i^{**}\) is the solution to Model~\eqref{model:qp} with \(\mathcal{R} \gets \mathcal{R}_l^{k+1}\), and set  
\[
\boldsymbol{v}^{**} = \boldsymbol{x}_l - \boldsymbol{x}_{i^{**}}.
\]
Since  
\[
\texttt{conv}(\{\boldsymbol{x}_{i^*}, \boldsymbol{x}_{i^{etm}}\}) \subseteq \texttt{conv}(\mathcal{R}_l^{k+1}),
\]
it follows that  
\[
\|\boldsymbol{v}^{**}\| \leq \|\boldsymbol{v}'\|.
\]
Therefore
\begin{equation*}
d(\boldsymbol{x}_l,\mathcal{R}_l^{k+1}) = \|\boldsymbol{v}^{**}\| < \|\boldsymbol{v}^{*}\| = d(\boldsymbol{x}_l,\mathcal{R}_l^k) 
\end{equation*}
Thus, by induction,  
\begin{equation*}
\lim_{k \to \infty} d(\boldsymbol{x}_l,\mathcal{R}_l^k) = 0.
\end{equation*}
Therefore, the while loop in lines 7-13 of Algorithm~\ref{algorithm:ConstructingConvexHull} converges.
\end{proof}
\vspace{10pt}

As the iteration progresses, \(\mathcal{R}_l\) grows by incorporating \(\boldsymbol{x}_{i^{etm}}\) at line 9, resulting in a rapid decrease in the size of the corresponding exterior set \[
\mathcal{A}^{\mathrm{ex}}(\mathcal{R}_l)
= \bigl\{\, \boldsymbol{x}_i \in \mathcal{A} \;\big|\; \boldsymbol{x}_i \notin \operatorname{conv}(\mathcal{R}_l) \,\bigr\}.
\]
Instead of decreasing by one element at a time, this reduction is anticipated to follow a more pronounced pattern, i.e.,  
\[
|\mathcal{A}^{\mathrm{ex}}(\mathcal{R}_l^{k+1})| - |\mathcal{A}^{\mathrm{ex}}(\mathcal{R}_l^k)| \gg 1
\]
for large \( n \). This accelerated contraction implies that the procedure in lines 7-13 may converge within a finite number of iterations. A visual representation of the monotonic decrease in the distance \( d(\boldsymbol{x}_l, \mathcal{R}_l^k) \) can be found in Appendix~\ref{app:monodecrease}.

\begin{figure*}[t]
    \centering
    \begin{subfigure}[b]{0.495\textwidth}
        \centering
        \includegraphics[width=\textwidth]{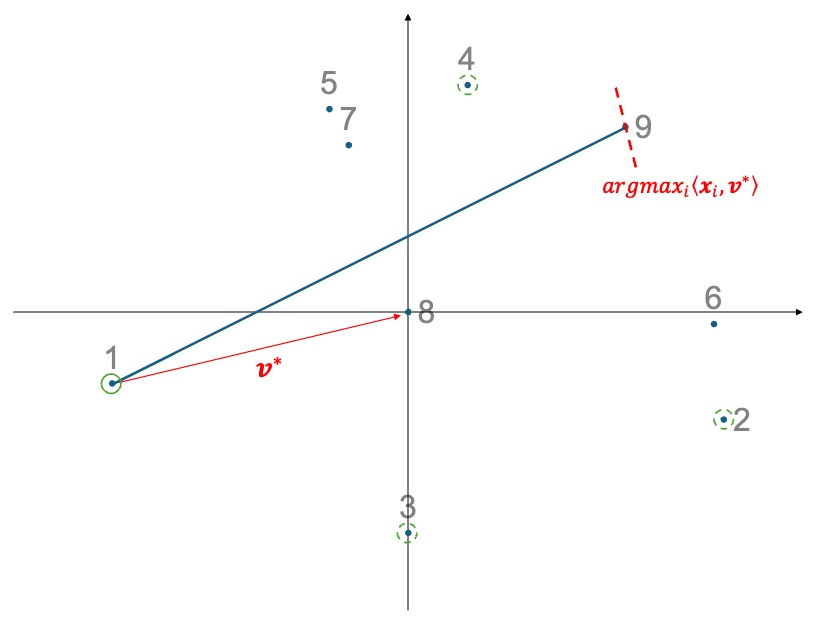}
        \caption{}
        \label{subfig:converg01}
    \end{subfigure}
    \begin{subfigure}[b]{0.495\textwidth}
        \centering
        \includegraphics[width=\textwidth]{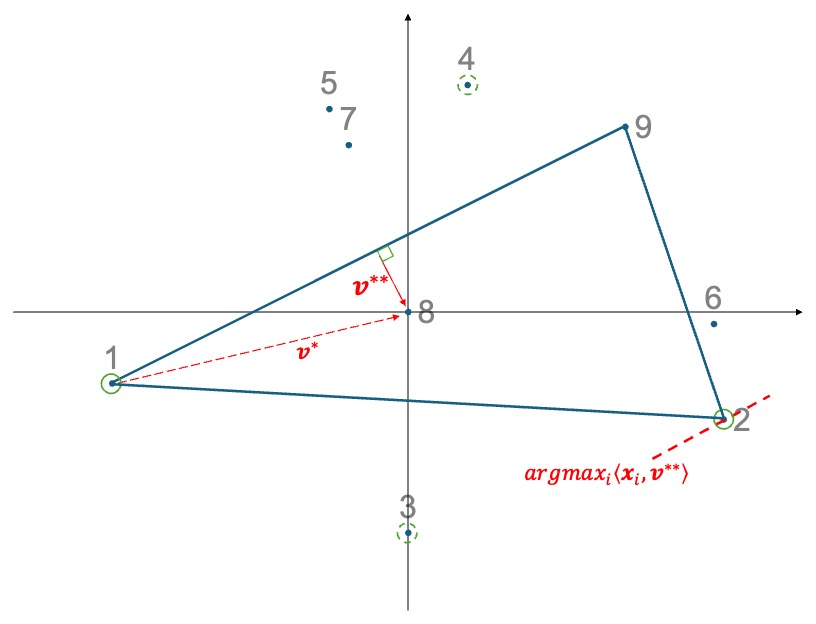}
        \caption{}
        \label{subfig:converg02}
    \end{subfigure}
    \caption{Implementation for Point~$\mathsf{8}$}
    \label{fig:converge}
\end{figure*}

Figure~\ref{fig:converge} illustrates the convergence process within a finite number of iterations during the processing of \( \boldsymbol{x}_8 \). For clarity, assume without loss of generality that \( \mathcal{R}_8 \) is initialized as \( \{\boldsymbol{x}_1\} \). In this scenario,  
\[
\boldsymbol{x}_9 = \operatorname*{\textit{arg}\,\textit{max}}_{i \in \mathcal{I}_\mathcal{A}} \langle \boldsymbol{x}_i, \boldsymbol{v}^* \rangle,
\]  
followed by  
\[
\boldsymbol{x}_2 = \operatorname*{\textit{arg}\,\textit{max}}_{i \in \mathcal{I}_\mathcal{A}} \langle \boldsymbol{x}_i, \boldsymbol{v}^{**} \rangle,
\]  
which yields  
\[
d(\boldsymbol{x}_8, \mathcal{R}_8) = 0,
\]
where  
\[
\mathcal{R}_8 = \{\boldsymbol{x}_1, \boldsymbol{x}_9, \boldsymbol{x}_2\}.
\]

This convergence behavior is applicable to any \( \boldsymbol{x}_l \in \mathcal{A} \), regardless of the initially selected \( \mathcal{R}_l \).

\begin{corollary}  
\label{theorem:convexhullguaranteedbyconvergence}  
Given the convergence of the while loop in lines 7-13 of Algorithm~\ref{algorithm:ConstructingConvexHull}, the exact \(\texttt{conv}(\mathcal{E})\) is guaranteed to be obtained in line 17 of Algorithm~\ref{algorithm:ConstructingConvexHull}.  
\end{corollary}  

\begin{proof}[Proof]
Since the while loop has converged,  
\[
d(\boldsymbol{x}_l, \mathcal{R}_l) = 0, \quad \forall l \in \mathcal{I}_\mathcal{A},
\]
which implies  
\begin{equation}
\boldsymbol{x}_l \in \texttt{conv}(\mathcal{R}_l), \quad \forall l \in \mathcal{I}_\mathcal{A}.
\label{eq:x_in_convRl}
\end{equation}
From line 15 of Algorithm~\ref{algorithm:ConstructingConvexHull}, it follows that  
\begin{equation}
\mathcal{R}_l \subseteq \mathcal{E}', \quad \forall l \in \mathcal{I}_\mathcal{A}.
\label{eq:R_subset_Eprime}
\end{equation}  
Taking the convex hull on both sides of \eqref{eq:R_subset_Eprime} yields  
\begin{equation}
\texttt{conv}(\mathcal{R}_l) \subseteq \texttt{conv}(\mathcal{E}').
\label{eq:convRl_subset_convEprime}
\end{equation}
Combining \eqref{eq:x_in_convRl} and \eqref{eq:convRl_subset_convEprime} leads to the conclusion that  
\[
\boldsymbol{x}_l \in \texttt{conv}(\mathcal{E}'), \quad \forall l \in \mathcal{I}_\mathcal{A}.
\]  
This confirms that all \(\boldsymbol{x}_l \in \mathcal{A}\) are enclosed within \(\texttt{conv}(\mathcal{E}')\). Consequently,  
\begin{equation*}
\texttt{conv}(\mathcal{E}') = \texttt{conv}(\mathcal{E}),
\end{equation*}  
implying that the complete convex hull is constructed in line 17.  
\end{proof}

\section{Complexity Analysis}
The computational cost of Algorithm~\ref{algorithm:ConstructingConvexHull} primarily arises from the 
$\operatorname*{\textit{arg}\,\textit{max}}$ operation and the solution of Model~\eqref{model:qp}.

In line~9, the $\operatorname*{\textit{arg}\,\textit{max}}$ operation has a computational cost of $O(nm)$. 
Computing the inner product $\langle \boldsymbol{x}_i, \boldsymbol{v}^* \rangle$ for each $\boldsymbol{x}_i$ requires $O(m)$ operations, 
and since this must be performed for all $n$ points, the total cost for computing all inner products is 
$n \cdot O(m) = O(nm)$. 
Finding the maximum among these $n$ scalar values adds an additional $O(n)$ cost, 
which is asymptotically dominated by $O(nm)$. 
Therefore, the overall complexity of the $\operatorname*{\textit{arg}\,\textit{max}}$ operation is $O(nm)$.

For each \( \boldsymbol{x}_l \), solving Model~\eqref{model:qp} with \( \mathcal{R}_l \) as the reference set requires 
\( O(|\mathcal{R}_l|^{p} \log(1/\epsilon)) \).
Suppose that for each \( \boldsymbol{x}_l \), the set \( \mathcal{R}_l \) is randomly initialized with one extreme point. 
In the worst case, the condition \( d(\boldsymbol{x}_l, \mathcal{R}_l) = 0 \) is satisfied only when 
\( |\mathcal{R}_l| = |\mathcal{E}| \). 
Let \( h = |\mathcal{E}| \) denote the number of extreme points. 
Then, the total computational effort required to construct \( \mathcal{R}_l \) can be expressed as
\begin{align*}
\sum_{|\mathcal{R}_l|=1}^{h} ( &O(|\mathcal{R}_l|^{p} \log(1/\epsilon)) + O(nm) )
= \\
&\quad O\!\left( \sum_{|\mathcal{R}_l|=1}^{h} |\mathcal{R}_l|^{p} \log(1/\epsilon) \right)
+ h \cdot O(nm).
\end{align*}
Using the standard asymptotic relation for power sums,
\[
O\!\left(\sum_{|\mathcal{R}_l|=1}^{h} |\mathcal{R}_l|^{p} \log(1/\epsilon) \right) = O(h^{p+1}\log(1/\epsilon)),
\]
the total per-point complexity becomes \( O(h^{p+1}\log(1/\epsilon) + hnm) \).
Aggregating over all \( n \) points yields
\[
O\big(nh^{p+1}\log(1/\epsilon) + n^2 h m\big).
\]
In the worst case, where \( h = n \), the total complexity simplifies to
\[
O(n^{p+2}\log(1/\epsilon) + n^3 m) = O(n^{p+2}\log(1/\epsilon)),
\]
since \( O(n^{p+2}\log(1/\epsilon)) \) dominates for sufficiently large \( n \).

\section{Summary}
This study develops a polynomial-time algorithm for the computation of exact convex hulls in high-dimensional spaces.

Section 2 reviews the convex hull problem, and presents the lemma for identifying convex hull vertices (extreme points).

Section 3 presents the algorithm that guarantees convergence to the exact convex hull. For each input point, an iterative reference-set update strategy is employed: at each iteration, the associated QP problem incorporates a newly identified extreme point into its reference set until convergence is achieved (i.e., the optimal objective value reaches zero). Once all points have been processed, the complete convex hull is obtained.

Section~4 provides a detailed analysis of the worst-case time complexity of the proposed algorithm, demonstrating that it is polynomial and independent of the problem dimension.

As demonstrated, the method retains full accuracy with polynomial complexity in the dataset size \( n \), presenting potential to reduce computational time and space for convex hull identification in high-dimensional spaces.

\section*{Appendix}
\begin{appendix}

\subsection{Proof of Lemma \ref{lemma:extreme point}}
\label{lemmaproof}

\begin{proof}[Proof]
Consider the optimization problem represented in Model~\eqref{eq:extremproof}, where $\boldsymbol{s}_l$ is a vector of slacks with $s_{jl}$ as its $j$-th element:

\begin{equation}
\label{eq:extremproof}
\begin{alignedat}{2}
    \alpha_l = & \min_{\lambda_i, s_{jl}} \sum_{j=1}^{m} |s_{jl}| \\
    & \text{s.t.} \\
    & \quad \boldsymbol{s}_l + \sum_{i \neq l} \lambda_i \boldsymbol{x}_i = \boldsymbol{x}_l \quad && \text{(}\ref{eq:extremproof}\text{a)} \\
    & \quad \sum_{i \neq l} \lambda_i = 1 \quad && \text{(}\ref{eq:extremproof}\text{b)} \\
    & \quad \lambda_i \geq 0, \forall i \quad && \text{(}\ref{eq:extremproof}\text{c)}
\end{alignedat}
\end{equation}
\\
Let \(\lambda_i^*\) for \(i \neq l\) and \(\boldsymbol{s}_l^{*} = (s_{1l}^*, \dots, s_{ml}^*)\) denote the optimal solution to Model~\eqref{eq:extremproof}, \(\alpha_l = \sum_{j=1}^{m} |s_{jl}^*|\). If \(\boldsymbol{x}_l\) cannot be represented as a convex combination of the other points, then \(\alpha_l > 0\) holds.\\ 
Expanding condition \text{(}\ref{eq:extremproof}\text{a)} component-wise yields
\[
s_{jl} + \sum_{i \neq l} \lambda_i x_{ji} = x_{jl}, \quad j=1, \dots, m.
\]
Multiplying each equation by \(v_j \in \mathbb{R}\) and summing over \(j = 1, \dots, m\) yields
\[
\sum_{j=1}^{m} v_j s_{jl} + \sum_{i \neq l} \lambda_i \sum_{j=1}^{m} v_j x_{ji} = \sum_{j=1}^{m} v_j x_{jl}.
\]
Let $\boldsymbol{v} = (v_1, \dots, v_m)^T \neq \boldsymbol{0}$. Using the inner product notation, the above can be expressed as:
\begin{equation}
\label{eq:inner-product}
    \langle \boldsymbol{v}, \boldsymbol{s}_l \rangle + \sum_{i \neq l} \lambda_i \langle \boldsymbol{v}, \boldsymbol{x}_i \rangle = \langle \boldsymbol{v}, \boldsymbol{x}_l \rangle.
\end{equation}
If $\boldsymbol{x}_l$ uniquely satisfies condition~\eqref{condition:argmax}, then by the constraints \text{(}\ref{eq:extremproof}\text{b)} and \text{(}\ref{eq:extremproof}\text{c)}:
\begin{equation*}
    \sum_{i \neq l} \lambda_i \langle \boldsymbol{v}, \boldsymbol{x}_i \rangle < \langle \boldsymbol{v}, \boldsymbol{x}_l \rangle.
\end{equation*}
It follows that:
\begin{equation*}
    \sum_{i \neq l} \lambda_i^* \langle \boldsymbol{v}, \boldsymbol{x}_i \rangle < \langle \boldsymbol{v}, \boldsymbol{x}_l \rangle.
\end{equation*}
To satisfy Equation~\eqref{eq:inner-product}, it follows that:
\begin{equation*}
    \langle \boldsymbol{v}, \boldsymbol{s}_l^* \rangle > 0.
\end{equation*}
Thus, $\boldsymbol{s}_l^* \neq \boldsymbol{0}$, and consequently:
\begin{equation*}
    \alpha_l = \sum_{j=1}^{m} |s_{jl}^*| > 0.
\end{equation*}
This proves that $\boldsymbol{x}_l$ is an extreme point if it uniquely satisfies condition~\eqref{condition:argmax}.
\end{proof}

\subsection{Monotonic Distance Decrease}
\label{app:monodecrease}

As shown in Figure~\ref{fig:convergent}, consider the processing of Point $\mathsf{c}$, whose reference set consists of Point $\mathsf{a}$ and Point $\mathsf{b}$. Solving Model~\eqref{model:qp} determines the virtual point \( i^* \) along with the vector \( \boldsymbol{v}^* \), which is perpendicular to the line through Point $\mathsf{a}$ and Point $\mathsf{b}$. This implies that the distance of Point $\mathsf{c}$ to the convex hull generated by its corresponding reference set \( \mathcal{R}_\mathsf{c} \) decreases. In the subsequent iteration, applying the $\operatorname*{\textit{arg}\,\textit{max}}$ operation with vector \( \boldsymbol{v}^{**} \) ensures that Point $\mathsf{c}$ becomes enclosed within \( \texttt{conv}(\mathcal{R}_\mathsf{c}) \), ultimately reducing the distance to zero.

Applying the $\operatorname*{\textit{arg}\,\textit{max}}$ operation shifts this line in the direction of \( \boldsymbol{v}^* \), identifying Point $\mathsf{d}$ as an extreme point. Adding Point $\mathsf{d}$ to the reference set results in a new virtual Point \( i^{**} \) and an updated vector \( \boldsymbol{v}^{**} \), which is perpendicular to the line through Point $\mathsf{b}$ and Point $\mathsf{d}$. Importantly, the magnitude of \( \boldsymbol{v}^{**} \) is smaller than that of \( \boldsymbol{v}^{*} \), demonstrating the decreasing distance property.

\begin{figure}[ht]
    \centering
    \includegraphics[width=0.44\textwidth]{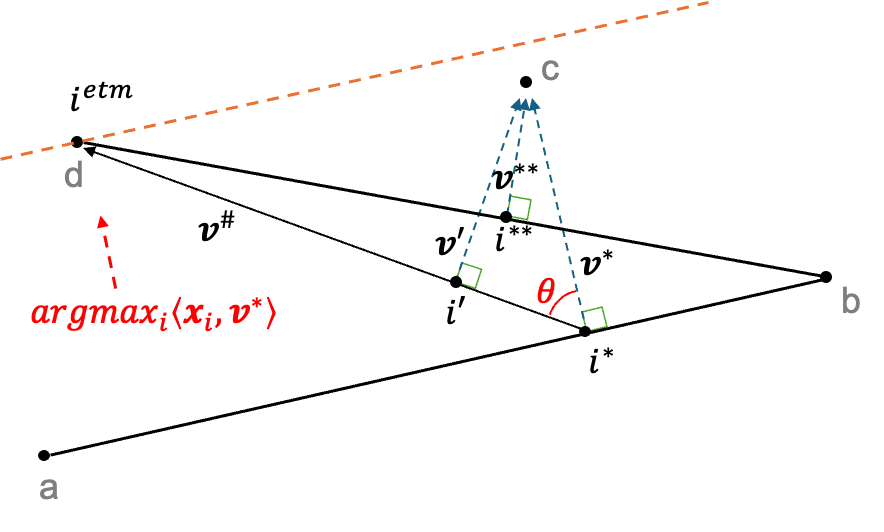}
    \caption{Monotonic Distance Decrease}
    \label{fig:convergent}
\end{figure}

\end{appendix}

\bibliography{references}

\end{document}